\documentclass[conference]{IEEEtran}  	
\usepackage{geometry}
\usepackage{scrextend}
\usepackage{graphicx}	
\usepackage{amssymb}
\usepackage{amsthm}
\usepackage[cmex10]{amsmath}
\usepackage{bm}

\newtheorem{thm}{Theorem}[section]
\newtheorem{lem}[thm]{Lemma}
\newtheorem{cor}[thm]{Corollary}
\theoremstyle{definition}
\newtheorem{defn}[thm]{Definition}

\title{Algebraic Conditions for Generating Accurate Adjacency Arrays}
\author{\IEEEauthorblockN{Karia Dibert$^{1,2}$, Hayden Jansen$^{1,2}$, Jeremy Kepner$^{1,2,3}$}
\IEEEauthorblockA{$^1$MIT Lincoln Laboratory, Lexington, MA \\ $^2$MIT Mathematics Department, Cambridge, MA \\ $^3$MIT Computer Science \& AI Laboratory, Cambridge, MA}
}
\begin{document}
\maketitle

\begin{abstract}

Data processing systems impose multiple views on data as it is processed by the system.  These views include spreadsheets, databases, matrices, and graphs. Associative arrays unify and simplify these different approaches into a common two-dimensional view of data. Graph construction, a fundamental operation in the data processing pipeline, is typically done by multiplying the incidence array representations of a graph, $\mathbf{E}_\mathrm{in}$ and $\mathbf{E}_\mathrm{out}$, to produce an adjacency matrix of the graph that can be processed with a variety of machine learning clustering techniques. This work focuses on establishing the mathematical criteria to ensure that the matrix product  $\mathbf{E}_\mathrm{out}^\intercal\mathbf{E}_\mathrm{in}$ is the adjacency array of the graph.  It will then be shown that these criteria are also necessary and sufficient for the remaining nonzero product of incidence arrays, $\mathbf{E}_\mathrm{in}^\intercal\mathbf{E}_\mathrm{out}$ to be the adjacency matrices of the reversed graph.  Algebraic structures that comply with the criteria will be identified and discussed.
\end{abstract}

\section{Introduction}
As data moves through a processing system, it is viewed from different perspectives by different parts of the system. Data often are first parsed into a tabular spreadsheet form, then ingested into database tables, analyzed with matrix mathematics, and presented as graphs of relationships. Often, the majority of time spent in building a data processing system is in the interfaces between the various steps.  These interfaces require a conversion from one mathematical perspective on the data to another \cite{kepner3}. Fortunately, spreadsheets, databases, matrices, and graphs share a common mathematical structure: they all use two-dimensional data structures in which each data element can be specified with a triple denoted by a row, column, and value. This structure is encompassed by the associative array. By using the associative array as a common mathematical abstraction across all steps, the construction time of a data processing system can be reduced \cite{kepner3}.
\let\thefootnote\relax\footnotetext{This work is sponsored by the Assistant Secretary of Defense for Research and Engineering under Air Force Contract \#FA8721-05-C-0002.  Opinions, interpretations, recommendations and conclusions are those of the authors and are not necessarily endorsed by the United States Government.
}

\begin{figure}[htb]
\includegraphics[width=3.4in]{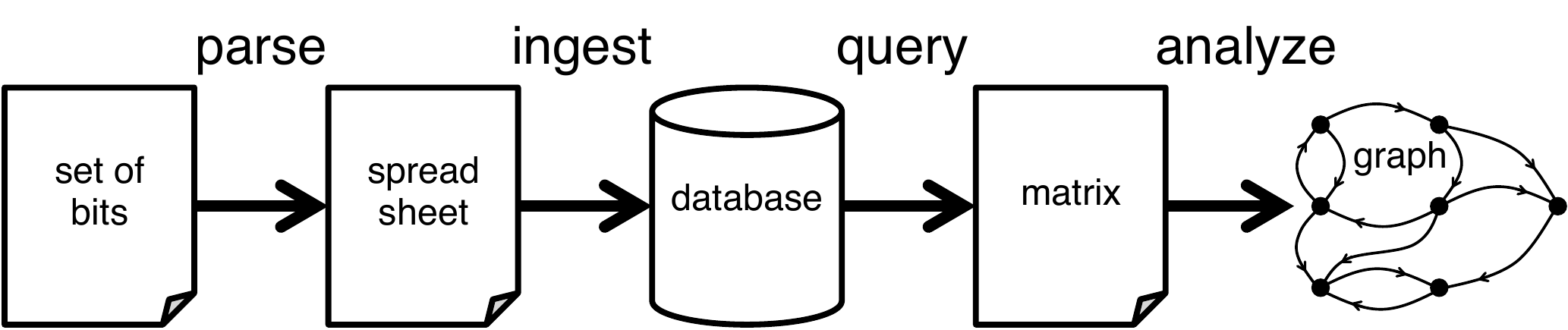}
\caption{The standard steps in a data processing system often require different perspectives on the data. Associative arrays enable a common mathematical perspective to be used across all the steps.}
\label{fig:pipeline}
\end{figure}

Associative arrays derive much of their power from their ability to represent data intuitively in easily understandable tables.
Two properties of associative arrays in particular are different from other two-dimensional arrangements of data.  First, each row label (or row key) and each column label (or column key) in an array is unique and sortable, thus allowing rows and columns to be queried efficiently. Secondly, because associative arrays contain no rows or columns that are entirely empty,  insertion, selection, and deletion of data can be performed by element-wise addition, element-wise multiplication, and array multiplication.

\subsection{Definitions} 

\begin{defn}
{\bf Associative Array}.  In \cite{kepner}, a 2-dimensional \textit{associative array} is defined as a map $\mathbf{A}: K_1 {\times} K_2 \to V$ from two finite, strict, totally ordered keysets to a value set with a commutative semi-ring structure.
\end{defn}

This paper will \textit{not} assume a full commutative semi-ring structure for the value set. It will be required only that the value set $V$ be closed under two binary operations and that $V$ contains the identity elements of these operations, denoted $0$ and $1$ respectively.

Associative arrays can represent graphs through both incidence arrays and adjacency arrays. The incidence arrays $\mathbf{E}_\mathrm{out}$ and $\mathbf{E}_\mathrm{in}$ of a graph map directed edges to vertices in the graph. 

\begin{defn}
{\bf Incidence Arrays}.
Suppose $G$ is a (possibly directed) weighted graph with edge set $K$ and vertex set $K_\mathrm{out} \cup K_\mathrm{in}$ where $K_\mathrm{out}$ is the set of vertices with positive out degree and $K_\mathrm{in}$ is the set of vertices with positive in degree.  

Then we say that the associative arrays $\mathbf{E}_\mathrm{out}: K{\times} K_\mathrm{out} \rightarrow V$ and $\mathbf{E}_\mathrm{in}: K {\times} K_\mathrm{in} \rightarrow V$ are the \emph{source incidence array} and \emph{target incidence array} of $G$ respectively if the following are true for all $k \in K$ and $a \in K_\mathrm{out} \cup K_\mathrm{in}$:

1.  $\mathbf{E}_\mathrm{out}(k, a)\neq 0$ if and only if edge $k$ is directed outward from vertex $a$ 
 
2.  $\mathbf{E}_\mathrm{in}(k, a)\neq 0$ if and only if edge $k$ is directed into vertex $a$.

Unless otherwise specified, $\mathbf{E}_\mathrm{out}$ and $\mathbf{E}_\mathrm{in}$ denote fixed source and target incidence arrays of a weighted, directed graph $G$.
\end{defn}

\begin{defn}
{\bf Adjacency Array}.
Given a graph $G$ with vertex set $K_1 \cup K_2$, we say that an associative array $\mathbf{A}: K_1 {\times} K_2 \to V$ is an \emph{adjacency array} of the graph $G$ if, for $a,b \in K_1 \cup K_2$, $\mathbf{A}(a,b) \neq 0$ if and only if $a$ is adjacent to $b$.  In a directed graph, this means that $\mathbf{A}(a,b) \neq 0$ if and only if there is an edge directed out of vertex $a$ and into vertex $b$.
\end{defn}

Incidence arrays are often readily obtained from raw data. In many cases, an associative array representing a spreadsheet or database table is already in the form of an incidence array. In order to analyze  a graph, it is often convenient to represent the graph as an adjacency array.  
Constructing an adjacency array from data stored in an incidence array via array multiplication is one of the most common and important steps in a data processing pipeline.

Multiplication of associative arrays is defined as in \cite{kepner}, where for associative arrays $\mathbf{A}: K_1{\times} K_2\rightarrow V$, $\mathbf{B}:K_2{\times} K_3 \rightarrow V$, and $\mathbf{C}:K_1{\times} K_3\rightarrow V$, the statement $\mathbf{C}=\mathbf{A} {\oplus}.{\otimes} \mathbf{B}$ means that $\mathbf{C}(i,j)=\bigoplus\limits_{k \in K} \mathbf{A}(i,k) \otimes \mathbf{B}(k,j)$. The array product $\mathbf{A} {\oplus}.{\otimes} \mathbf{B}$ will henceforth be denoted $\mathbf{AB}$ for simplicity.

Given a graph $G$ with vertex set $K_\mathrm{in}\cup K_\mathrm{out}$ and edge set $K$, the construction of adjacency arrays for $G$ relies on the assumption that $\mathbf{E}_\mathrm{out}^\intercal\mathbf{E}_\mathrm{in}$ is an adjacency array of $G$.  This assumption is certainly true in the most common case where the value set is composed of non-negative reals and the operations $\oplus$ and $\otimes$ are arithmetic plus ($+$) and arithmetic times (${\times}$) respectively. However, one hallmark of associative arrays is their ability to contain as values non-traditional data. For these value sets, $\oplus$ and $\otimes$ may be redefined to operate on non-numerical values. 

For example, \cite{kepner} suggests the value set be the set of all alphanumeric strings, with $\oplus = max()$ and $\otimes = min()$. It is not immediately apparent in this case whether $\mathbf{E}_\mathrm{out}^\intercal\mathbf{E}_\mathrm{in}$ is an adjacency array of the graph whose set of vertices is $K_\mathrm{out} \cup K_\mathrm{in}$. The focus of this paper is to determine the criteria on the value set $V$ and the operations $\oplus$ and $\otimes$ so that the above assumption is always true.

\subsection{Problem Statement} 
 Let $G$ be a graph with source and target incidence arrays $\mathbf{E}_\mathrm{out}$ and $\mathbf{E}_\mathrm{in}$, respectively.  An entry of $\mathbf{E}_\mathrm{out}^\intercal\mathbf{E}_\mathrm{in}$, called $\mathbf{E}_\mathrm{out}^\intercal\mathbf{E}_\mathrm{in}(x,y)$ for $x \in K_\mathrm{out}$, and $y \in K_\mathrm{in}$ may be expressed as
\[\mathbf{E}_\mathrm{out}^\intercal\mathbf{E}_\mathrm{in}(x,y)= \bigoplus\limits_{k \in K} \mathbf{E}_\mathrm{out}(k,x) \otimes \mathbf{E}_\mathrm{in}(k,y)\]
For $\mathbf{E}_\mathrm{out}^\intercal\mathbf{E}_\mathrm{in}$ to be the adjacency array of $G$, the entry $\mathbf{E}_\mathrm{out}^\intercal\mathbf{E}_\mathrm{in}(x,y)$ must be nonzero if and only if $x$ is adjacent to $y$, which is identical to saying that the entry must be nonzero if and only if the following is true: 
\[\exists k \in K \text{ so that } \mathbf{E}_\mathrm{out}(k,x) \neq 0 \text{ and } \mathbf{E}_\mathrm{in}(k,y) \neq 0\]
Then, the problem ultimately becomes: Under what conditions is the following statement true for all graphs with edge set $K$, vertex set $K_\mathrm{out} \cup K_\mathrm{in}$, and weights in a set $V$ that is closed under two binary operations, $\oplus$ and $\otimes$, and contains the identity elements of these operations, denoted $0$ and $1$ respectively?
\begin{multline} 
\bigoplus_{k \in K} \mathbf{E}_\mathrm{out}(k,x) \otimes \mathbf{E}_\mathrm{in}(k,y) \neq 0 \iff \\ 
\exists k \in K \text{ so that } \mathbf{E}_\mathrm{out}(k,x) \neq 0  \text{ and } \mathbf{E}_\mathrm{in}(k,y) \neq 0 \label{problem}\end{multline}

The conditions under which this statement is true will now be established.

\section{Criteria}
The criteria on the set $V$ and the operations $\oplus$ and $\otimes$ such that the above statement is true are:

\begin{addmargin}{4em}
1. $V$ contains no non-trivial additive inverses:

$\nexists v, w \in V, \; v,w\neq 0$ such that $v \oplus w = 0$

\noindent 2. $V$ satisfies the zero product property:

$\forall a, b \in V, \; a,b \neq 0 \Rightarrow a \otimes b \neq 0$

\noindent 3. Zero annihilates $V$

$\forall v \in V, \; v \otimes 0=0 \otimes v =0$
\end{addmargin}
\begin{thm}
The above criteria are necessary, given (\ref{problem}).
\label{necessary}
\end{thm}

\begin{proof}
Assume (\ref{problem}). This is equivalent to
\begin{multline}
\bigoplus\limits_{k \in K} \mathbf{E}_\mathrm{out}(k,x) \otimes \mathbf{E}_\mathrm{in}(k,y) = 0 \iff \\ 
\nexists k \in K \, \mbox{so that} \, \mathbf{E}_\mathrm{out}(k,x) \neq 0 \text{ and } \mathbf{E}_\mathrm{in}(k,y) \neq 0 \end{multline}
which in turn is equivalent to
\begin{multline}
\bigoplus\limits_{k \in K} \mathbf{E}_\mathrm{out}(k,x) \otimes \mathbf{E}_\mathrm{in}(k,y) = 0 \iff \\
\forall k \in K, \mathbf{E}_\mathrm{out}(k,x) = 0 \text{ or } \mathbf{E}_\mathrm{in}(k,y) = 0
\end{multline}
This expression may be split up into two conditional statements
\begin{multline} 
\bigoplus\limits_{k \in K} \mathbf{E}_\mathrm{out}(k,x) \otimes \mathbf{E}_\mathrm{in}(k,y) = 0 \Rightarrow \\
 \forall k \in K, \mathbf{E}_\mathrm{out}(k,x) = 0 \text{ or }  \mathbf{E}_\mathrm{in}(k,y) = 0 \label{leftright}
 \end{multline}
and
\begin{multline} 
\forall k \in K, \mathbf{E}_\mathrm{out}(k,x) = 0 \, \mbox{or} \, \mathbf{E}_\mathrm{in}(k,y) = 0 \Rightarrow \\
 \bigoplus\limits_{k \in K} \mathbf{E}_\mathrm{out}(k,x) \otimes \mathbf{E}_\mathrm{in}(k,y) = 0 \label{rightleft}
 \end{multline}

\begin{lem}
The first criterion, that $V$ contain no non-trivial additive inverses, is necessary for (\ref{leftright}).
\end{lem}
\begin{proof}
\renewcommand{\qedsymbol}{}
Suppose there exist nonzero $v,w \in V$ such that $v \oplus w = 0$, or that non-trivial additive inverses exist.  Then we may choose  our graph $G$ to have edge set $\{k_1,k_2\}$ and vertex set $\{a,b\}$, where both $k_1,k_2$ start from $a$ and end at $b$.  Then defining $\mathbf{E}_\mathrm{out}(k_1,a)=v$, $\mathrm{E}_\mathrm{out}(k_2,a)=w$, and $\mathrm{E}_\mathrm{in}(k_i,b)=1$ give us proper source and target incidence arrays for $G$.  Moreover, we find
\begin{equation*}
\mathbf{E}_\mathrm{out}^\intercal\mathbf{E}_\mathrm{in}(b,a)=(v\otimes 1)\oplus (w\otimes 1)=v\oplus w=0
\end{equation*}
which contradicts (\ref{leftright}).  Therefore, no such non-zero $v$ and $w$ may be present in $V$, meaning it is necessary that $V$ contain no non-trivial additive inverses.
\end{proof}

\begin{lem}
The second criterion, that $V$ satisfies the zero product property, is necessary for (\ref{leftright}).
\end{lem}
\begin{proof}
\renewcommand{\qedsymbol}{}
Suppose there exist nonzero $v, w \in V$ such that $ v \otimes w = 0$.  Then we may choose our graph $G$ to have edge set $\{k\}$ and vertex set $\{a\}$, with a single self-loop given by $k$.  Then defining $\mathbf{E}_\mathrm{out}(k,a)=v$ and $\mathbf{E}_\mathrm{in}(k,a)=w$ give us source and target incidence arrays for $G$.  Moreover, we find
\begin{equation*}
\mathbf{E}_\mathrm{out}^\intercal\mathbf{E}_\mathrm{in}(a,a)=\mathbf{E}_\mathrm{out}(k,a)\otimes \mathbf{E}_\mathrm{in}(k,a)=v\otimes w= 0
\end{equation*}
which contradicts (\ref{leftright}). Therefore, no such $v$ and $w$ may be present in $V$, so $V$ satisfies the zero-product property.
\end{proof}

\begin{lem}
The third criterion, that the additive identity must annihilate $V$, is necessary for (\ref{rightleft}).
\end{lem}
\begin{proof}
\renewcommand{\qedsymbol}{}
Suppose there exists non-zero $v \in V$ such that $v \otimes 0 \neq 0$ or $0 \otimes v \neq 0$. Suppose $v\otimes 0 \neq 0$.  Then we may choose our graph $G$ to have edge set $\{k\}$ and vertex set $\{a,b\}$, with a single self-loop at $a$ given by $k$.  Then defining $\mathbf{E}_\mathrm{out}(k,a)=v=\mathbf{E}_\mathrm{in}(k,a)$ give us source and target incidence arrays for $G$.  Moreover, we find
\begin{equation*}
\mathbf{E}_\mathrm{out}^\intercal\mathbf{E}_\mathrm{in}(a,b)=\mathbf{E}_\mathrm{out}(k,a)\otimes \mathbf{E}_\mathrm{in}(k,b)=v\otimes 0 \neq 0
\end{equation*}
which contradicts (\ref{rightleft}), and so we must have $v\otimes 0=0$.  By considering $\mathbf{E}_\mathrm{out}^\intercal\mathbf{E}_\mathrm{in}(b,a)$, we can show that $0\otimes v=0$ as well.
\end{proof}
\end{proof}

\begin{thm} \label{sufficient}
The above criteria are sufficient for (\ref{problem}).
\end{thm}
\begin{proof}
Assume that zero is an annhilator, that there are no non-trivial additive inverses in $V$, and that $V$ satisfies the zero-product property.

\noindent Nonexistence of non-trivial additive inverses and the zero product property give
\begin{multline}
\exists k \in K \text{ so that } \mathbf{E}_\mathrm{out}(k,x) \neq 0 \text{ and } \mathbf{E}_\mathrm{in}(k,y) \neq 0 \Rightarrow \\
\bigoplus\limits_{k \in K} \mathbf{E}_\mathrm{out}(k,x) \otimes \mathbf{E}_\mathrm{in}(k,y) \neq 0 
\end{multline}
which is the contrapositive of (\ref{leftright}). And, that zero is an annihilator gives
\begin{multline}
\forall k \in K, \mathbf{E}_\mathrm{out}(e,x) = 0 \text{ or } \mathbf{E}_\mathrm{in}(e,y) = 0 \Rightarrow \\
\bigoplus\limits_{k \in } \mathbf{E}_\mathrm{out}(k,x) \otimes \mathbf{E}_\mathrm{in}(k,y) = 0 
\end{multline}
which is (\ref{rightleft}). 
As (\ref{leftright}) and (\ref{rightleft}) combine to form (\ref{problem}), it is established that the conditions are sufficient for (\ref{problem}).
\end{proof}

\subsection{Corollary for Remaining Product}

The remaining product of the incidence arrays that may be nonzero is $\mathbf{E}_\mathrm{in}^\intercal\mathbf{E}_\mathrm{out}$. The above requirements will now be shown to be necessary and sufficient for the remaining product to be the adjacency array of the reverse of the graph. For the following corollaries, let $G$ be a graph with incidence matrices $\mathbf{E}_\mathrm{out}$ and $\mathbf{E}_\mathrm{in}$.

\begin{cor}
The established requirements are necessary and sufficient so that $\mathbf{E}_\mathrm{in}^\intercal\mathbf{E}_\mathrm{out}$ is the adjacency matrix of the reverse of $G$. 
\label{revcor}
\end{cor}
\begin{proof}
\renewcommand{\qedsymbol}{}
Let $R(G)$ denote the reverse of $G$, and let $\bm{\mathcal{E}}_\mathrm{out}$ and $\bm{\mathcal{E}}_\mathrm{in}$ be source and target incidence arrays for $R(G)$, respectively.  Recall that $R(G)$ is defined to have the same edge and vertex sets as $G$ but changing the directions of the edges, i.e. if an edge $k$ leaves a vertex $a$ in $G$, then it enters $a$ in $R(G)$, and vice-a-versa.  As such, $\mathbf{E}_\mathrm{out}(k,a)\neq 0$ if and only if $\bm{\mathcal{E}}_\mathrm{in}\neq 0$, and likewise $\mathbf{E}_\mathrm{in}(k,a)\neq 0$ if and only if $\bm{\mathcal{E}}_\mathrm{out}\neq 0$.  As such, choosing $\mathbf{E}_\mathrm{out}=\bm{\mathcal{E}}_\mathrm{in}$ and $\mathbf{E}_\mathrm{in}=\bm{\mathcal{E}}_\mathrm{out}$ give valid target and source incidence matrices for $R(G)$, respectively.  

Then by \ref{sufficient} we find that $\bm{\mathcal{E}}_\mathrm{out}^\intercal\bm{\mathcal{E}}_\mathrm{in}=\mathbf{E}_\mathrm{in}^\intercal\mathbf{E}_\mathrm{out}$.
\end{proof}
 
\section{Implications and Future Work}

It is now straightforward to identify algebraic structures that comply with the established criteria. Notably, all semi-rings without additive inverses and which satisfy the zero product property comply, such as $\mathbb{N}$ or $\mathbb{R}_{\geq 0}$ with the standard addition and multiplication as well as any linearly-ordered set with $\oplus$ and $\otimes$ given by $\max$ and $\min$, respectively.  Some non-examples, however, include the max-plus algebra or non-trivial Boolean algebras, which do not satisfy the zero-product property, or rings, which except for the zero ring have non-trivial additive inverses.  Furthermore, the value sets of associative arrays need not be defined exclusively as semi-rings, as several semi-ring-like structures satisfy the criteria. These structures may lack the properties of additive or multiplicative commutativity, additive or multiplicative associativity, or distributivity of multiplication over addition, which are not necessary to ensure that the product of incidence arrays yields an adjacency array.

\subsection{Note About Special Cases}
The criteria guarantee an accurate adjacency array for any dataset that satisfies them, regardless of value distribution in the incidence arrays. However, if the incidence arrays are known to possess a certain structure, it is possible to circumvent some of the conditions and still always produce adjacency arrays. For example, if each keyset of an undirected incidence array $\mathbf{E}$ is a list of documents and the array entries are sets of words shared by documents, then it is necessary that a word in $\mathbf{E}(i,j)$ and $\mathbf{E}(m,n)$ has to be in $\mathbf{E}(i,n)$ and $\mathbf{E}(m,j)$. This structure means that when multiplying $\mathbf{E}^\intercal\mathbf{E}$ using $\oplus = \cup$ and $\otimes = \cap$, a nonempty set will never be ``multiplied'' by (intersected with) a disjoint nonempty set. This eliminates the need for the zero-product property to be satisfied, as every multiplication of nonempty sets is already guaranteed to produce a nonempty set. The array produced will contain as entries a list of words shared by those two documents.

\subsection{Future Work}
Though the criteria ensure that the product of incidence arrays will be an adjacency array, they do not ensure that certain matrix properties hold. For example, the property $(\mathbf{AB})^\intercal=\mathbf{B}^\intercal\mathbf{A}^\intercal$ may be violated under these criteria, as $(\mathbf{E}_\mathrm{out}^\intercal\mathbf{E}_\mathrm{in})^\intercal$ is not necessarily equal to $\mathbf{E}_\mathrm{in}^\intercal\mathbf{E}_\mathrm{out}$. (For this matrix transpose property to always hold, the operation $\otimes$ would have to be commutative.) Future work would involve considering more matrix properties and assembling a stronger set of conditions to ensure that these properties hold. This would enable users to utilize such properties when performing analysis on data stored in associative arrays.

\end{document}